\documentclass[letterpaper, conference]{IEEEtran}
\usepackage{amsmath,amssymb,amsthm,cite}
\usepackage{graphicx}
\usepackage{color}
\usepackage{subcaption}
\usepackage[left=0.75in,right=0.75in, top=0.75in,bottom=0.75in]{geometry}
\renewcommand{\IEEEtitletopspaceextra}{0.25in}
\usepackage{tkz-euclide}
\usetkzobj{all}
\usepackage{tikz}
\usepackage{url}
\usepackage{flushend}
\usepackage{todonotes}
\usepackage[bookmarks=false,colorlinks=true,urlcolor=blue,citecolor=blue,linkcolor=blue]{hyperref}
\usepackage{cleveref}
\newtheorem{lem}{Lemma}
\newtheorem{thm}{Theorem}
\newtheorem{defn}{Definition}
\newtheorem{coro}{Corollary}
\newtheorem{clm}{Claim}

\newcommand{\E}[1]{\mathbb{E}\left[{#1}\right]}

\newcommand{\posfunc}[1]{\lvert {#1}\rvert ^{+}}
\newcommand{\HyperExp}{\textit{HyperExp}}
\newcommand{\SExp}{\textit{ShiftedExp}}

\graphicspath{{Figures/}}

\crefname{equation}{}{}
\Crefname{equation}{}{}
\crefname{thm}{theorem}{theorems}
\Crefname{thm}{Theorem}{Theorems}
\crefname{clm}{claim}{claims}
\Crefname{clm}{Claim}{Claims}
\Crefname{coro}{Corollary}{Corollaries}
\Crefname{lem}{Lemma}{Lemmas}
\Crefname{sec}{Section}{Sections}
\crefname{app}{appendix}{appendices}
\Crefname{app}{Appendix}{Appendices}
\crefname{prop}{proposition}{propositions}
\Crefname{prop}{Proposition}{Propositions}
\Crefname{propty}{Property}{Properties}
\crefname{figure}{fig.}{figures}
\Crefname{figure}{Fig.}{Figures}
\crefname{defn}{definition}{definitions}
\Crefname{defn}{Definition}{Definitions}
\crefname{fact}{fact}{facts}
\Crefname{fact}{Fact}{Facts}
\crefname{appendix}{appendix}{appendices}
\Crefname{appendix}{Appendix}{Appendices}
\crefname{algo}{algorithm}{algorithms}
\Crefname{algo}{Algorithm}{Algorithms}
\crefname{algorithm}{algorithm}{algorithms}
\Crefname{algorithm}{Algorithm}{Algorithms}
\crefname{conj}{conjecture}{conjectures}
\Crefname{conj}{Conjecture}{Conjectures}
\crefname{obs}{observation}{observations}
\Crefname{obs}{Observation}{Observations}

\begin{document}
\newcommand{\totalCaches}{m}
\newcommand{\nFork}{n}
\newcommand{\nPartialFork}{r}
\newcommand{\kJoin}{k}
\newcommand{\xm}{x_m}


\title{Efficient Replication of Queued Tasks for \\Latency Reduction in Cloud Systems}
\author{
\IEEEauthorblockN{Gauri Joshi}
\IEEEauthorblockA{EECS Dept.,
MIT\\
Cambridge, MA 02139, USA \\
Email: gauri@mit.edu}
\and
\IEEEauthorblockN{Emina Soljanin}
\IEEEauthorblockA{
Bell Labs, Alcatel-Lucent \\
Murray Hill NJ 07974, USA\\
Email: emina@bell-labs.com}
\and
\IEEEauthorblockN{Gregory Wornell}
\IEEEauthorblockA{
EECS Dept., MIT \\
Cambridge MA 02139, USA\\
Email: gww@mit.edu}
}

\maketitle

\renewcommand{\thefootnote}{} 
\footnotetext[1]{This work was supported in part by NSF under Grant No. CCF-1319828, AFOSR under Grant No. FA9550-11-1-0183, and a Schlumberger Faculty for the Future Fellowship.} \renewcommand{\thefootnote}{\arabic{footnote}}
\vspace{-0.25cm}
\begin{abstract}
In cloud computing systems, assigning a job to multiple servers and waiting for the earliest copy to finish is an effective method to combat the variability in response time of individual servers. Although adding redundant replicas always reduces service time, the total computing time spent per job may be higher, thus increasing waiting time in queue. The total time spent per job is also proportional to the cost of computing resources. We analyze how different redundancy strategies, for eg.\ number of replicas, and the time when they are issued and canceled, affect the latency and computing cost. We get the insight that the \emph{log-concavity} of the service time distribution is a key factor in determining whether adding redundancy reduces latency and cost. If the service distribution is log-convex, then adding maximum redundancy reduces both latency and cost. And if it is log-concave, then having fewer replicas and canceling the redundant requests early is more effective.

\end{abstract}

\section{Introduction}
\label{sec:intro}
\subsection{Motivation}

An increasing number of applications are now hosted on the cloud. Some examples are streaming (NetFlix, YouTube), storage (Dropbox, Google Drive) and computing (Amazon EC2, Microsoft Azure) services. 
A major advantage of cloud computing and storage is that the large-scale sharing of resources provides scalability and flexibility. A side-effect of the sharing of resources is the variability in the latency experienced by the user due to queueing, pre-emption by other jobs with higher priority, server outages etc. The problem becomes further aggravated when the user is executing a job with several parallel tasks on the cloud, because the slowest task becomes the bottleneck in job completion. Thus, ensuring seamless, low-latency service to the end-user is a challenging problem in cloud systems. 

One method to reduce latency that has gained significant attention in recent years is the use of redundancy. In cloud computing, executing a task on multiple machines and waiting of one to finish can significantly reduce the latency \cite{dean_tail_2013}. Similarly, in cloud storage systems requests to access the content can be assigned to multiple replicas,  such that it is only sufficient to download one replica. This can help reduce latency significantly. 

However, redundancy can result in increased use of resources such as computing time, and network bandwidth. In frameworks as Amazon EC2 and Microsoft Azure which offer computing as a service, the server time spent is proportional to the money spent in renting the machines. In this work we aim to understand this trade-off between latency and computing cost and propose scheduling policies that can achieve a good trade-off. Our analysis also results in some fundamental advances in the analysis of queues with redundant requests.

\begin{table*}[t]
\footnotesize
\centering
\caption{Optimal redundancy strategies when the service time is log-concave or log-convex. `Canceling redundancy early' means that we cancel redundant tasks when any $1$ task reaches the head of its queue, instead of waiting for it to be served. 
\label{tbl:result_summary}}
\begin{tabular}{| p{2.9cm} | p{4.0cm} | p{2.0cm} | p{2.5 cm} | p{2.45 cm} | }
\hline
 & \multicolumn{2}{c|} { \textbf{Log-concave service time} } & \multicolumn{2}{c|} { \textbf{Log-convex service time} }\\
\hline
 & \textbf{ Latency-optimal} & \textbf{Cost-optimal} & \textbf{Latency-optimal} & \textbf{Cost-optimal} \\
\hline
\vspace{0.05cm} \textbf{Cancel redundancy early or keep it?} \vspace{0.05cm}  
& \vspace{0.05cm} {Low load:~Keep~Redundancy, \,\, High load:~Cancel early} \vspace{0.05cm} & \vspace{0.05cm} Cancel early & \vspace{0.05cm} Keep Redundancy & \vspace{0.05cm} Keep Redundancy\\
\hline
\vspace{0.05cm} \textbf{Partial forking to $r$ out of $n$ servers} 
&  \vspace{0.05cm} {Low load:~$r=n$ (fork to all), \,\,\, High load:~$r = 1$ (fork to one)} \vspace{0.1cm} & \vspace{0.05cm} $r=1$ & \vspace{0.05cm}  $r=n$ & \vspace{0.05cm} $r=n$ \\
\hline
\end{tabular}
\end{table*}

\subsection{Previous Work}

\textit{\textbf{Systems Work}}:
One of the earliest instances of exploiting redundancy to reduce latency is the use of multiple routing paths \cite{maxemchuk2} to send packets in networks. 
A similar idea has also been recently studied in \cite{vulimiri_low_2013}. In large-scale cloud computing frameworks such as MapReduce \cite{map_reduce}, the slowest tasks of a job (stragglers) become a bottleneck in its completion. Several recent works in systems such as  \cite{ananthanarayanan_effective_2013, zaharia_sparrow} explore straggler mitigation techniques where redundant replicas of straggling tasks are launched to reduce latency. 

Although the use of redundancy has been explored in systems literature, there is little work on the rigorous analysis of how it affects latency, and in particular the cost of resources. We now review some of that work. 

\textit{\textbf{Exponential Service Time}}:
In distributed storage systems, erasure coding can be used to store a content file on $n$ servers such that it can be recovered by accessing any $k$ out of the $n$ servers. Thus download latency can be reduced by forking each request to all $n$ servers and waiting for any $k$ to respond. In \cite{gauri_yanpei_emina_allerton, gauri_yanpei_emina_jsac} we found bounds on the expected latency using the $(n,k)$ fork-join model with exponential service time. This is a generalization of the $(n,n)$ fork-join system, which was actively studied in queueing literature \cite{flatto1984two, nelson_tantawi}. In recent years, there is a renewed interest in fork-join queues due to their application to distribution computing frameworks such as MapReduce. Another related model with a centralized queue instead of queues at each of the $n$ servers was analyzed in \cite{mds_queue}. Most recently, \cite{gardner_sigmetrics_2015} presents an analysis of latency with heterogeneous job classes for the replicated ($k=1$) case with exponential service time. 

\textit{\textbf{General Service Time}}:
Few practical systems have exponentially distributed service time. For example, studies of download time traces from Amazon S3 \cite{docomo_1, docomo_2} indicate that the service time is not exponential in practice, but instead a shifted exponential. For service time distributions that are `new-worse-than-used' \cite{cao_nbu_1991}, it is shown in \cite{koole_righter_2008} that it is optimal to fork a job to maximum number of servers. The choice of scheduling policy for new-worse-than-used (NWU) and new-better-than-used (NBU) distributions is also studied in \cite{righter_job_rep_2009, shah_when_2013, sun_shroff}. The NBU and NWU notions are closely related to the log-concavity of service time studied in this work. 


\textit{\textbf{The Cost of Redundancy}}:
If the service time is assumed to be exponential, then adding redundancy does not cause any increase in cost of computing time. But since the exponential assumption does not generally hold true in practice, it is important to determine the cost of using redundancy. Simulation results with non-zero fixed cost of removal of redundant requests is considered in \cite{shah_when_2013}. The total server time spent on each job is considered in \cite{wang_efficient_2014, sigmetrics_arxiv_2015} for a distributed system without queueing of requests. In \cite{gauri_mama_2015} we presented an analysis of the latency and cost of the $(n,k)$ fork-join with and without early cancellation of redundant tasks.

\subsection{Our Contributions}

In this work, we consider a general service time distribution, unlike exponential service time assumed in many previous works. We analyze the impact of redundancy on the latency, and also the computing cost (total server time spent per job). Incidentally, our computing cost metric serves as a powerful tool to compare different redundancy strategies in the high traffic regime. 

The analysis gives the insight that the log-concavity (log-convexity) of the tail distribution $\bar{F}_X$ of service time is a key factor in determining when redundancy helps. Here are some instances, that are also summarized in Table~\ref{tbl:result_summary}. For example, a redundancy strategy is to fork each job to queues at $n$ servers, and wait for any one replica to finish. An alternate strategy is to cancel the redundant replicas as soon as any one reaches the head of its queue. We can show that early cancellation of redundancy can reduce both latency and cost for log-concave $\bar{F}_X$, but it is not effective for log-convex $\bar{F}_X$. In another instance, suppose we fork each job to only a subset $r$ out of the $n$ servers. Then we can show that forking to more servers (larger $r$) is always better for log-convex $\bar{F}_X$. But for log-concave $\bar{F}_X$, larger $r$ reduces latency only in the low traffic regime, and always increases the computing cost. 

\section{Problem Formulation}
\label{sec:prob_setup}
\subsection{Fork-Join Model and its Variants}
\label{subsec:sys_model}

Consider a distributed system with $n$ statistically identical servers. We define the $(n,1)$ fork-join system as follows.

\begin{defn}[$(n,1)$ fork-join system]
\label{defn:fork_join}
Each incoming job is forked into $\nFork$ tasks that join first-come first-serve queues at the $\nFork$ servers. When any one task is served, all remaining tasks are canceled and abandon their queues immediately. 
\end{defn}

The term `task' refers to a replica of the job. This is a special case of the $(n,k)$ fork-join system considered in \cite{gauri_yanpei_emina_allerton, gauri_yanpei_emina_jsac} where any $k$ out of $n$ tasks are sufficient to complete the job. General $k>1$ arise in approximate computing, or in content download from erasure coded distributed storage. Fig.~\ref{fig:fork_join_queue} illustrates the $(3,1)$ fork-join system. 
%


\begin{figure}[t]
\centering
\includegraphics[width=3.2in]{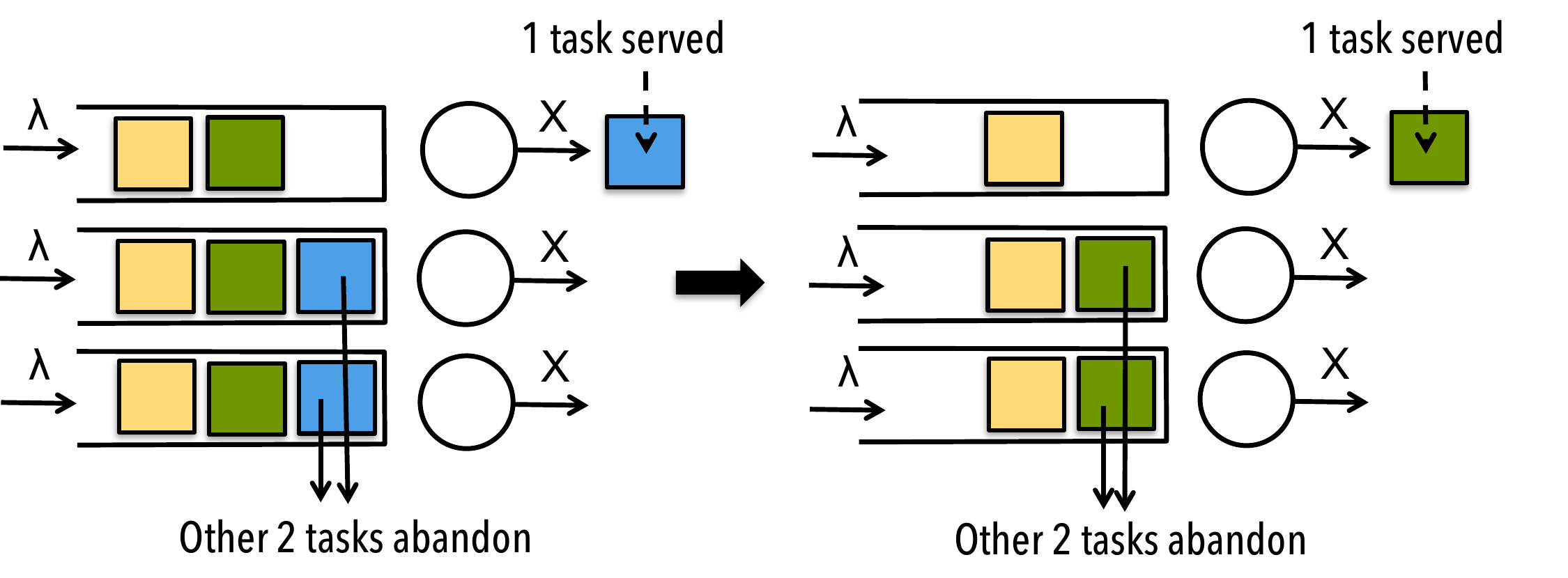}
\caption{The $(3,1)$ fork-join system. When any $1$ out of $3$ tasks of a job is served, the remaining $2$ tasks abandon their queues immediately.\label{fig:fork_join_queue}}
\end{figure}

%
%


Instead of waiting for any one task to finish, we could cancel the redundant tasks early, when any task starts service. A similar idea has been proposed in systems work \cite{zaharia_sparrow}. We refer to this variant as the $(n,1)$ fork-early cancel system defined formally as follows. 

\begin{figure}[t]
\centering
\includegraphics[width=2.9in]{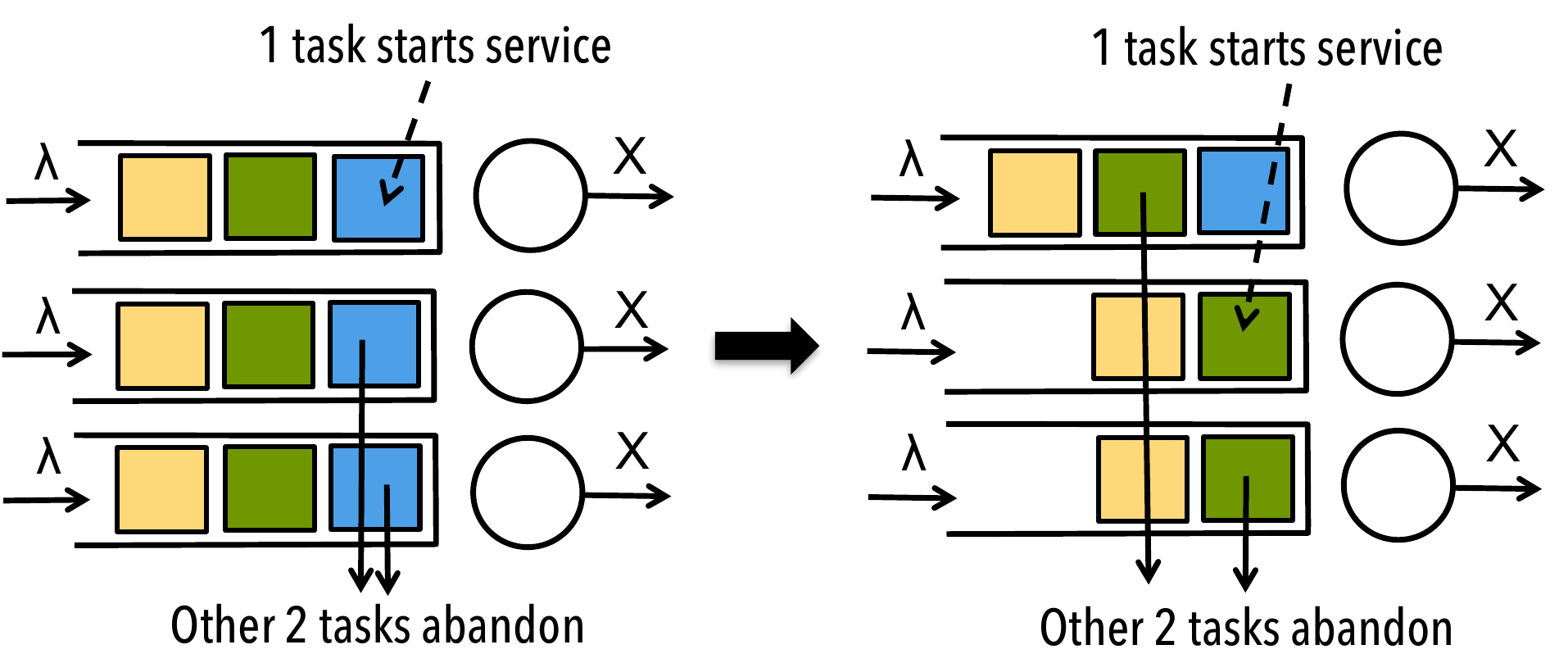}
\caption{The $(3,1)$ fork-early cancel system. When any $1$ out of $3$ tasks of a job starts service, the others abandon their queues. \label{fig:fork_early_cancel}}
\end{figure}

\begin{defn}[$(n,1)$ fork-early cancel system]
Each incoming job is forked to the $n$ servers. When any task starts service, we cancel the redundant tasks immediately. If more than one tasks start service simultaneously, we preserve any one task chosen uniformly at random. 
\end{defn}

Fig.~\ref{fig:fork_early_cancel} illustrates the $(3,1)$ fork-early cancel system. 
Early cancellation can save the total time spent per job (computing cost), but could result in an increase in latency because of loss of diversity. In Section~\ref{sec:rep_with_queueing} we compare the $(n,1)$ fork-join and the $(n,1)$ fork-early-cancel systems. 

Due to the network cost of issuing and canceling tasks, it may be prohibitively expensive to fork a job to all $n$ servers. Thus, we consider a partial forking variant defined as follows.

\begin{defn}[$(n, r, 1)$ partial fork-join system]
Each incoming job is forked into $r$ out of the $n$ servers. When any one task is served, the redundant tasks are canceled immediately and the job exits the system.
\end{defn}
 
The $r$ servers can be chosen according to different scheduling policies such as random, round-robin, least-work-left etc. Partial forking can save total computing time as well as the network cost, which is proportional to the number of servers each job is forked to. In \Cref{sec:partial_fork} we develop insights into the best choice of $r$ and the scheduling policy to achieve a good latency-cost trade-off. 

Other variants include a combination of partial forking and early cancellation, or delaying invocation of some of the redundant tasks. Although not considered here, our analysis techniques can be extended to these variants.

\subsection{Arrival and Service Distributions}
Consider that jobs arrive to the system at rate $\lambda$ per second, according to a Poisson process. The Poisson assumption is required only for the exact analysis and bounds of latency $\E{T}$ (defined below). The results for cost $\E{C}$, and the insights into choosing the best redundancy strategy hold for any arrival process. 

After a task of the job reaches the head of its queue, the time taken to serve it can be random due to various factors such as virtualization, disk seek time, server outages, pre-emption by other jobs etc. We model this task service time by the random variable $X > 0$, with cumulative distribution function $F_X(x)$ and assume that it is i.i.d.\ across requests and servers. Dependence of the service time on the job size can be modeled by adding a constant to $X$. For example, some recent work \cite{docomo_1, docomo_2} on analysis of content download from Amazon S3 observed that $X$ is shifted exponential, where $\Delta$ is proportional to the size of the content and the exponential part is the random delay in starting the data transfer.

We use $\bar{F}_X(x)$ to denote $\Pr(X>x)$, the tail distribution (inverse CDF) of $X$. We use $X_{1:n}$ to denote the smallest of $n$ i.i.d.\ random variables $X_1, X_2, \dots X_n$.

\subsection{Latency and Cost Metrics}
\label{subsec:perf_metrics}

We now define the metrics of the latency and resource cost whose trade-off is analyzed in the rest of the paper.

\begin{defn}[Latency]
The latency $\E{T}$ is defined as the expected time from when a job arrives, until when any one of its tasks is complete. 
\end{defn}


\begin{defn}[Computing Cost]
The expected computing cost $\E{C}$ is the expected total time spent by the servers serving a job, not including the time spent in the queue. 
\end{defn}

If a task is canceled before it reaches the head of its queue, the cost incurred at that server is zero. In computing-as-a-service frameworks, the expected computing cost is proportional to money spent on renting machines to run a job on the cloud. Although not analyzed explicitly in this paper, we note that there is also a network cost of issuing and canceling the redundant tasks, proportional to $n$ for the $(n,1)$ fork-join and fork-early-cancel system, and $r$ in the $(n,r,1)$ partial-fork-join system. 

\section{Preliminary Concepts}
\label{sec:key_concepts}
We present some preliminary concepts that are vital for understanding the results presented in the rest of the paper. 

\subsection{Using $\E{C}$ to Compare Systems}
\label{subsec:using_E_C}

Since the cost metric $\E{C}$ is the expected time spent by servers on each job, higher $\E{C}$ implies higher expected waiting time for subsequent jobs. Since the latency $\E{T}$ is dominated by the waiting time at high arrival rates, $\E{C}$ can be used to compare different redundancy policies in the high traffic regime. In particular, we can only compare policies that are symmetric across the servers, defined formally as follows.
\begin{defn}[Symmetric Policy]
\label{defn:symmetric_forking}
In a symmetric policy, the tasks of each job are forked one or more of the $n$ servers such that the expected task arrival rate is equal across the servers.
\end{defn}

Most commonly used policies: random, round-robin, shortest queue etc.\ are symmetric across the $n$ servers. In \Cref{clm:capacity_in_terms_of_EC}, we express the service capacity in terms on $\E{C}$. \Cref{coro:high_traffic_comp} then follows because higher service capacity implies lower latency in the high load regime.

\begin{clm}[Service Capacity in terms of $\E{C}$]
\label{clm:capacity_in_terms_of_EC}
For a system of $n$ servers with a symmetric redundancy policy, and any arrival process with rate $\lambda$, the service capacity, that is, the maximum $\lambda$ such that $\E{T} < \infty$ is
\begin{align}
\lambda_{max} &= \frac{n}{\E{C}} \label{eqn:capacity_in_terms_of_EC}
\end{align}
\end{clm}

\begin{proof}
For a symmetric policy, the mean time spent by each server per job is $\E{C}/n$. Thus the server utilization is $\rho = \lambda \E{C}/n$. To keep the system stable such that $\E{T} < \infty$, the server utilization must be less than $1$. The result in \eqref{eqn:capacity_in_terms_of_EC} follows from this.
\end{proof}

\begin{coro}
\label{coro:high_traffic_comp}
The symmetric redundancy strategy that results in a lower $\E{C}$, also gives lower $\E{T}$ in the high traffic regime, when $\lambda$ is close to the service capacity.  
\end{coro}



\subsection{Log-concavity of $\bar{F}_X$}
When the tail distribution $\bar{F}_X$ of service time is either `log-concave' or `log-convex', we get clear insights into how redundancy affects latency and cost. Log-concavity of $\bar{F}_X$ is defined formally as follows.

\begin{defn}[Log-concavity and log-convexity of $\bar{F}_X$]
\label{defn:log_concave}
The tail distribution $\bar{F}_X$ is said to be log-concave (log-convex) if $ \log \Pr(X>x)$ is concave (convex) in $x$ for all $x \in [0, \infty)$. 
\end{defn}

For brevity, when we say $X$ is log-concave (log-convex) in this paper, we mean that $\bar{F}_X$ is log-concave (log-convex). An interesting implication of log-concavity is that if $\bar{F}_X$ is log-concave, 
\begin{equation}
\Pr(X>x+t|X>t) \leq \Pr(X>x) \label{eqn:new_better_than_used}
\end{equation}
The inequality is reversed if $\bar{F}_X$ is log-convex\footnote{The definition of the notion `new-better-than-used' considered in \cite{koole_righter_2008} is same as \eqref{eqn:new_better_than_used}. Other names used to refer to new-better-than-used distributions are `light-everywhere' in \cite{shah_when_2013} and `new-longer-than-used' in \cite{sun_shroff}.}. Equality holds for the exponential distribution, which is both log-convex and log-concave. As a result the mean residual life $\E{X-t|X>t}$ decreases (increases) with the elapsed time $t$ if $\bar{F}_X$ is log-concave (log-convex).  

The numerical results in this paper use the shifted exponential, and hyper exponential as examples of log-concave and log-convex distributions respectively. The shifted exponential, denoted by $\SExp(\Delta,\mu)$ is an exponential with rate $\mu$, plus a constant shift $\Delta \geq 0$. The hyper-exponential distribution, denoted by $\HyperExp(\mu_1, \mu_2, p)$. It is a mixture of two exponentials with rates $\mu_1$ and $\mu_2$ where the exponential with rate $\mu_1$ occurs with probability $p$.

If we fork a job to all $r$ idle servers and wait for any $1$ copy to finish, the expected computing cost $\E{C} = r\E{X_{1:r}}$. \Cref{lem:r_E_X_1_r_trend} below gives how $r \E{X_{1:r}}$ varies with $r$ for log-concave (log-convex) $\bar{F}_X$. It is central to proving several key results in this paper. 

\begin{lem}
\label{lem:r_E_X_1_r_trend}
If $X$ is log-concave (log-convex), $r \E{X_{1:r}}$ is non-decreasing (non-increasing) in $r$.
\end{lem}

The proof of \Cref{lem:r_E_X_1_r_trend} is omitted here and can be found in the extended version \cite{gauri_tompecs_arxiv_2015}. 
We refer readers to \cite{log_concave} for other properties and examples of log-concave distributions.

\subsection{Relative Task Start Times}
\label{subsec:task_start_times}

The relative start times of the $n$ tasks of a job is an important factor affecting the latency and cost. Let the relative task start times  be $t_1 \leq t_2 \leq \cdots t_n$ where $t_1 = 0$ without loss of generality and $t_i$ for $i > 1$ are measured from the instant when the earliest task starts service. For instance, if $n=3$ tasks start at times $3$, $4$ and $7$, then $t_1 = 0$, $t_2 = 4-3 = 1$ and $t_3 = 7-3 = 4$ respectively. In the case of partial forking when only $r$ tasks are invoked, we can consider $t_{r+1}, \cdots t_n$ to be $\infty$. 

Let $S$ be the time from when the earliest task starts service, until any one task finishes. Thus it is minimum of $X_1+t_1 , X_2 +t_2, \cdots X_n + t_n$, where $X_i$ are i.i.d. with distribution $F_X$. The tail distribution of $S$ is given by
\begin{align}
\Pr(S> s) &= \prod_{i=1}^{n} \Pr(X > s - t_n) \label{eqn:S_tail_dist}
\end{align}

The computing cost $C$ is given by,
\begin{align}
C &= S + \posfunc{ S - t_2} + \cdots + \posfunc{S - t_n}. \label{eqn:C_expr}
\end{align}

The relative task start times $t_i$ affect $C$ in two opposing ways. The negative part of each term in \eqref{eqn:C_expr} increases with $t_i$, but the expected value of $S$ increases \eqref{eqn:S_tail_dist} with $t_i$. By analyzing \eqref{eqn:C_expr} we get several crucial insights in the rest of the paper. For instance, in Section~\ref{sec:partial_fork} we show that when $\bar{F}_X$ is log-convex, having $t_1 = t_2 = \cdots = t_n = 0$ gives the lowest $\E{C}$. Then using \Cref{clm:capacity_in_terms_of_EC} we can infer that it is optimal to fork a job to all $n$ servers when $\bar{F}_X$ is log-convex. 


\section{$(n,1)$ system with and without early cancellation} 
\label{sec:rep_with_queueing}
\pdfoutput=1
We now analyze the latency and cost of the $(n,1)$ fork-join and $(n,1)$ fork-early-cancel systems defined in Section~\ref{sec:prob_setup}. We get the insight that it is better to cancel redundant tasks early if $\bar{F}_X$ is log-concave and traffic is high. But if $\bar{F}_X$ is log-convex, retaining the redundant tasks is always better.


\begin{figure}[t]
\centering
\includegraphics[width=3.3in]{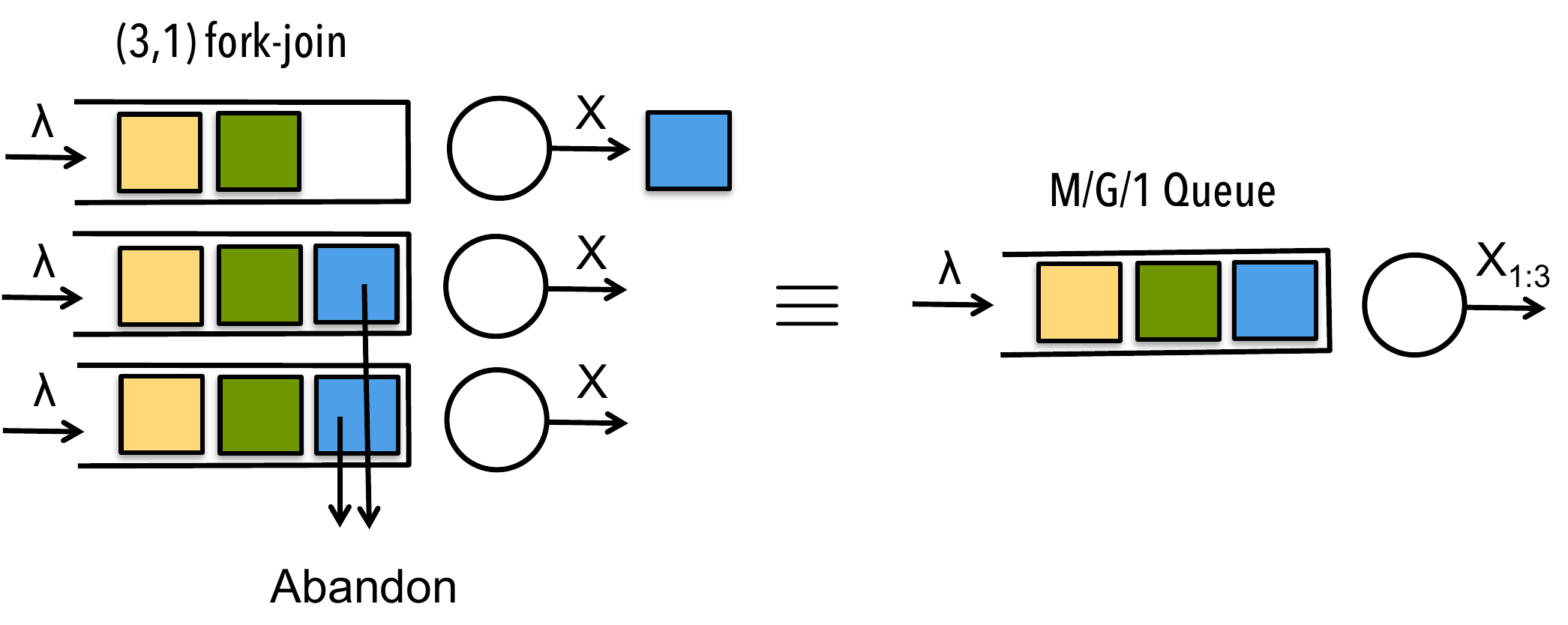}
\caption{Equivalence of the $(n,1)$ fork-join system with an $M/G/1$ queue with service time $X_{1:n}$, the minimum of $n$ i.i.d.\ random variables $X_1, X_2, \dots , X_n$. \label{fig:fork_join_mg1_eq}}
\vspace{-0.5cm}
\end{figure}

\subsection{Latency-Cost Analysis}

\begin{lem}
\label{lem:mg1_eq}
The latency $T$ of the $(n,1)$ fork-join system is equivalent in distribution to that of an $M/G/1$ queue with service time $X_{1:n}$. 
\end{lem}

\begin{proof}
Consider the first job that arrives to a $(n,1)$ fork-join system when all servers are idle. Thus, the $n$ tasks of this job start service at their respective servers simultaneously. The earliest task finishes after time $X_{1:n}$, and all other tasks are immediately. So, the tasks of all subsequent jobs arriving to the system also start simultaneously at the $n$ servers as illustrated in Fig.~\ref{fig:fork_join_mg1_eq}. Hence, arrival and departure events, and the latency of an $(n,1)$ fork-join system is equivalent in distribution to an $M/G/1$ queue with service time $X_{1:n}$.
\end{proof}

\begin{thm}
\label{thm:rep_queueing}
The expected latency and computing cost of an $(n,1)$ fork-join system are given by
\begin{align}
\E{T} &= \E{T^{M/G/1}} = \E{X_{1:n}} + \frac{\lambda \E{X_{1:n}^2}}{2(1 - \lambda \E{X_{1:n}})} \label{eqn:E_T_rep_queueing} \\
\E{C} &= n \cdot \E{X_{1:n}} \label{eqn:E_C_rep_queueing} 
\end{align}
where $X_{1:n} = \min (X_1, X_2, \dots , X_n)$ for i.i.d.\ $X_i \sim F_X$.
\end{thm}

\begin{proof}
By \Cref{lem:mg1_eq}, the latency of the $(n,1)$ fork-join system is equivalent in distribution to an $M/G/1$ queue with service time $X_{1:n}$. The expected latency of an $M/G/1$ queue is given by the Pollaczek-Khinchine formula \eqref{eqn:E_T_rep_queueing}. The expected cost $\E{C} = n \E{X_{1:n}}$ because each of the $n$ servers spends $X_{1:n}$ time on the job. This can also be seen by noting that $S = X_{1:n}$ when $t_i = 0$ for all $i$, and thus $C = n X_{1:n}$ in \eqref{eqn:C_expr} .
\end{proof}

In \Cref{coro:rep_queueing_ET} and \Cref{coro:rep_queueing_EC} we characterize how $\E{T}$ and $\E{C}$ vary with $n$. The behavior of $\E{C}$ follows from \Cref{lem:r_E_X_1_r_trend}.

\begin{coro}
\label{coro:rep_queueing_ET}
For the $(n,1)$ fork-join system with any service distribution $F_X$, the expected latency $\E{T}$ is non-increasing with $n$. 
\end{coro}
 
\begin{coro}
\label{coro:rep_queueing_EC}
If $\bar{F}_X$ is log-concave (log-convex), then $\E{C}$ is non-decreasing (non-increasing) in $n$.
\end{coro}

Fig.~\ref{fig:ET_vs_EC_rep_shifted_exp_var_n} and Fig.~\ref{fig:ET_vs_EC_rep_hyper_exp_var_n} show the expected latency versus cost for log-concave and log-convex $\bar{F}_X$, respectively. In Fig.~\ref{fig:ET_vs_EC_rep_shifted_exp_var_n}, the arrival rate $\lambda = 0.25$, and $X$ is shifted exponential $\SExp(\Delta, 0.5)$, with different values of $\Delta$. For $\Delta > 0$, there is a trade-off between expected latency and cost. Only when $\Delta = 0$, that is, $X$ is a pure exponential (which is generally not true in practice), we can reduce latency without any additional cost. In Fig.~\ref{fig:ET_vs_EC_rep_hyper_exp_var_n}, arrival rate $\lambda = 0.5$, and $X$ is hyperexponential $\HyperExp(0.4, 0.5, \mu_2)$ with different values of $\mu_2$. We get a simultaneous reduction in $\E{T}$ and $\E{C}$ as $n$ increases. The cost reduction is steeper as $\mu_2$ increases.

\begin{figure}[t]
        \centering
      \includegraphics[width=3.2in]{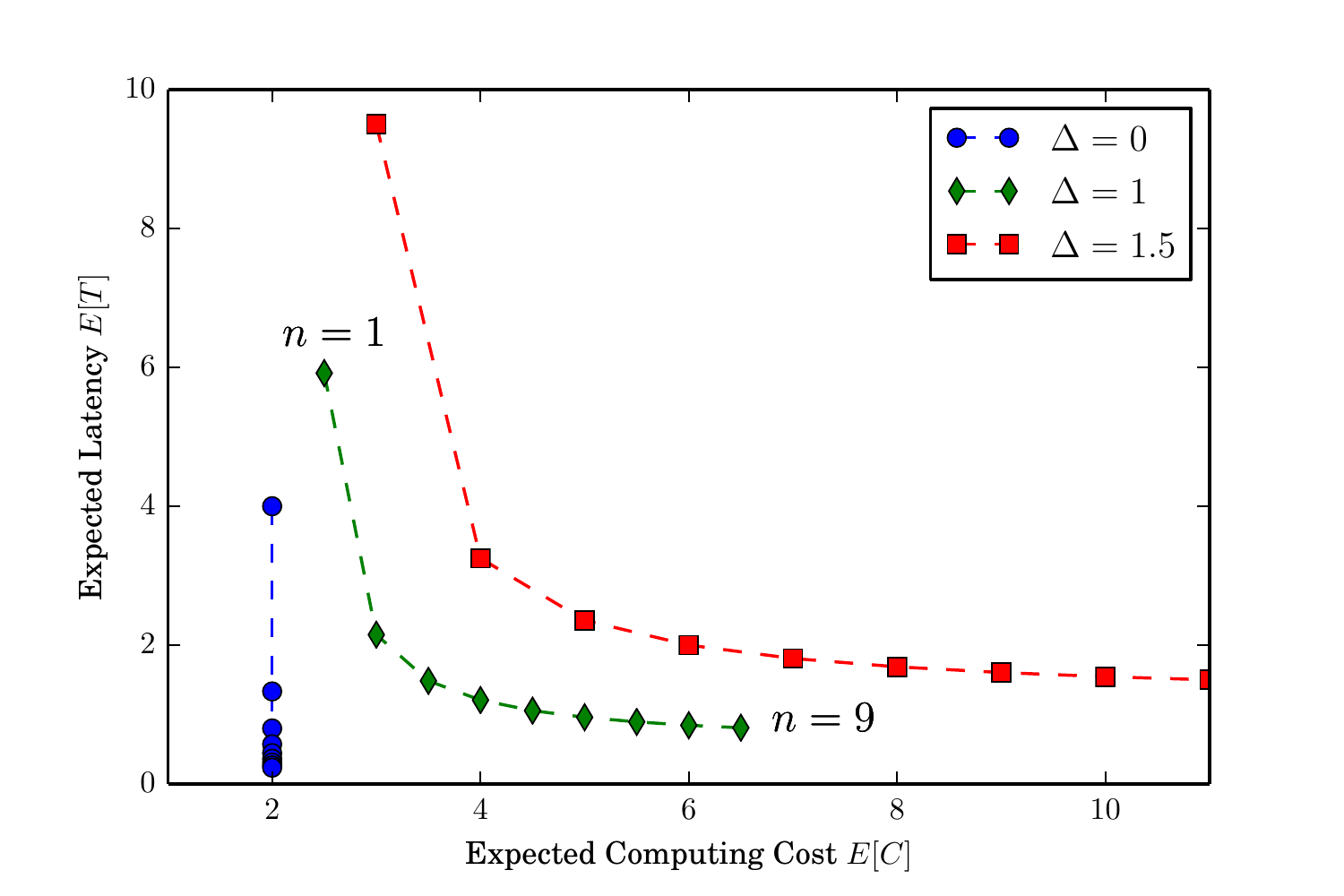}
      \caption{ The service time $X \sim \Delta + \text{Exp}(\mu)$ (log-concave), with $\mu = 0.5$, $\lambda = 0.25$. As $n$ increases along each curve, $\E{T}$ decreases and $\E{C}$ increases. Only when $\Delta =0$, latency reduces at no additional cost. \label{fig:ET_vs_EC_rep_shifted_exp_var_n}}
 \vspace{-0.2cm}
\end{figure}

\begin{figure}
\centering
\includegraphics[width=3.2in]{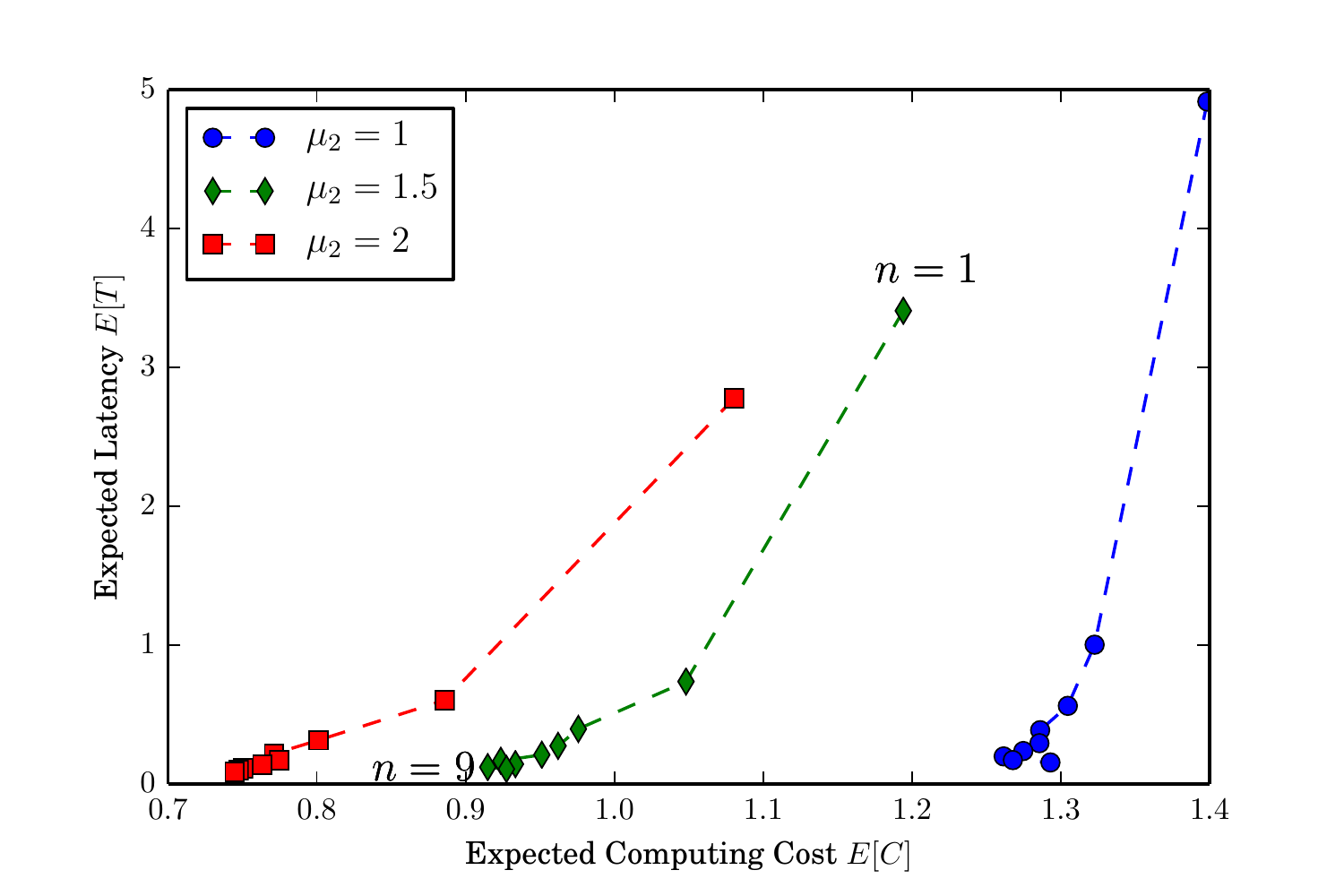}
\caption{ The service time $X \sim \HyperExp(0.4, \mu_1, \mu_2)$ (log-convex), with $\mu_1 = 0.5$, different values of $\mu_2$, and $\lambda = 0.5$.  Expected latency and cost both reduce as $n$ increases along each curve.\label{fig:ET_vs_EC_rep_hyper_exp_var_n}}
 \vspace{-0.2cm}
\end{figure}

\subsection{Early Task Cancellation}
\label{subsec:rep_early_cancel}
We now analyze the $(n, 1)$ fork-early-cancel system, where we cancel redundant tasks as soon as any task reaches the head of its queue. Intuitively, early cancellation can save computing cost, but the latency could increase due to the loss of diversity advantage provided by retaining redundant tasks. Comparing it to $(n,1)$ fork-join system, we gain the insight that early cancellation is better when $\bar{F}_X$ is log-concave, but ineffective for log-convex $\bar{F}_X$.

\begin{thm}
\label{thm:rep_queueing_early_cancel}
The expected latency and cost of the $(n,1)$ fork-early-cancel system are given by
\begin{align}
\E{T} &=  \E{T^{M/G/n}}, \\
\E{C} &= \E{X}, \label{eqn:rep_case_cost_early_cancel} 
\end{align}
where $T^{M/G/n}$ is the response time of an $M/G/n$ queueing system with service time $X \sim F_X$. 
\end{thm}

\begin{proof}
In the $(n,1)$ fork-early-cancel system, when any one tasks reaches the head of its queue, all others are canceled immediately. The redundant tasks help find the shortest queue, and exactly one task of each job is served by the first server that becomes idle. Thus, as illustrated in Fig.~\ref{fig:fork_early_mgn_eq}, the latency of the $(n,1)$ fork-early-cancel system is equivalent in distribution to an $M/G/n$ queue. Hence $\E{T} = \E{T^{M/G/n}}$ and $\E{C} = \E{X}$.
\end{proof}

\begin{figure}[t]
\centering
\includegraphics[width=3.3in]{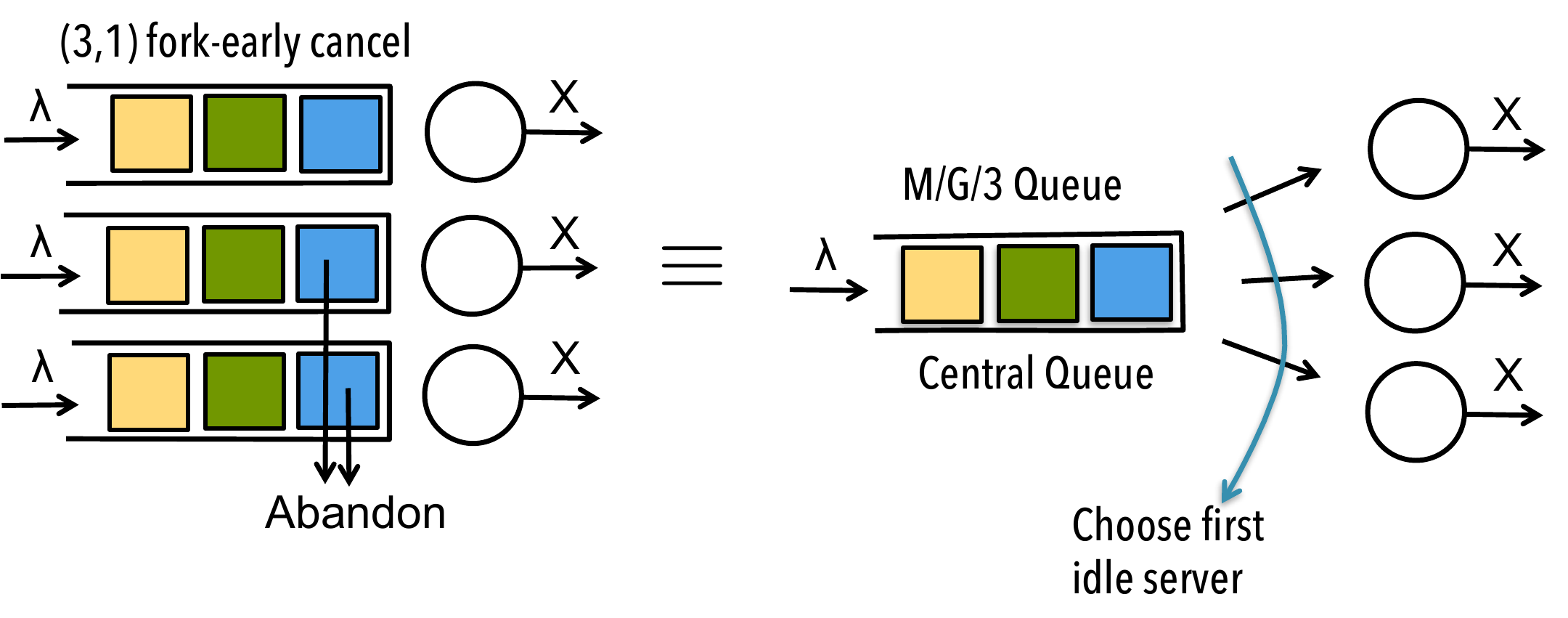}
\caption{Equivalence of the $(n,1)$ fork-early cancel system to an $M/G/n$ queue with each server taking time $X \sim F_X$ to serve task, i.i.d.\ across servers and tasks. \label{fig:fork_early_mgn_eq}}
 \vspace{-0.25cm}
\end{figure}

The exact analysis of mean response time $\E{T^{M/G/n}}$ has long been an open problem in queueing theory. A well-known approximation given by \cite{lee_loughton} is,
\begin{align}
 \E{T^{M/G/n}} \approx \E{X} +  \frac{ \E{X^2}}{2 \E{X}^2}\E{W^{M/M/n}}  \label{eqn:rep_case_latency_early_cancel}  
\end{align}
where $\E{W^{M/M/n}}$ is the expected waiting time in an $M/M/n$ queueing system with load $\rho = \lambda \E{X}/n$. It can be evaluated using the Erlang-C model \cite[Chapter~14]{mor_book}. 
%
We now compare the latency and cost with and without early cancellation given by Theorem~\ref{thm:rep_queueing_early_cancel} and Theorem~\ref{thm:rep_queueing}. \Cref{coro:early_cancel_E_C_trend} below follows from \Cref{lem:r_E_X_1_r_trend}.


\begin{coro}
\label{coro:early_cancel_E_C_trend}
If $\bar{F}_X$ is log-concave (log-convex), then $\E{C}$ of the $(n,1)$ fork-early-cancel system is greater than equal to (less than or equal to) that of $(n,1)$ fork-join system. 
\end{coro}

In the low $\lambda$ regime, the $(n,1)$ fork-join system gives lower $\E{T}$ than $(n,1)$ fork-early-cancel because of higher diversity due to redundant tasks. By \Cref{coro:high_traffic_comp}, the high $\lambda$ regime, the system with lower $\E{C}$ has lower expected latency. 

\begin{coro}
\label{coro:early_cancel_E_T_trend}
If $\bar{F}_X$ is log-concave, early cancellation gives higher $\E{T}$ than $(n,1)$ fork-join when $\lambda$ is small, and lower in the high $\lambda$ regime. If $\bar{F}_X$ is log-convex, then early cancellation gives higher $\E{T}$ for both low and high $\lambda$.
\end{coro}


\begin{figure}[t]
    \centering
    \includegraphics[width=3.2in]{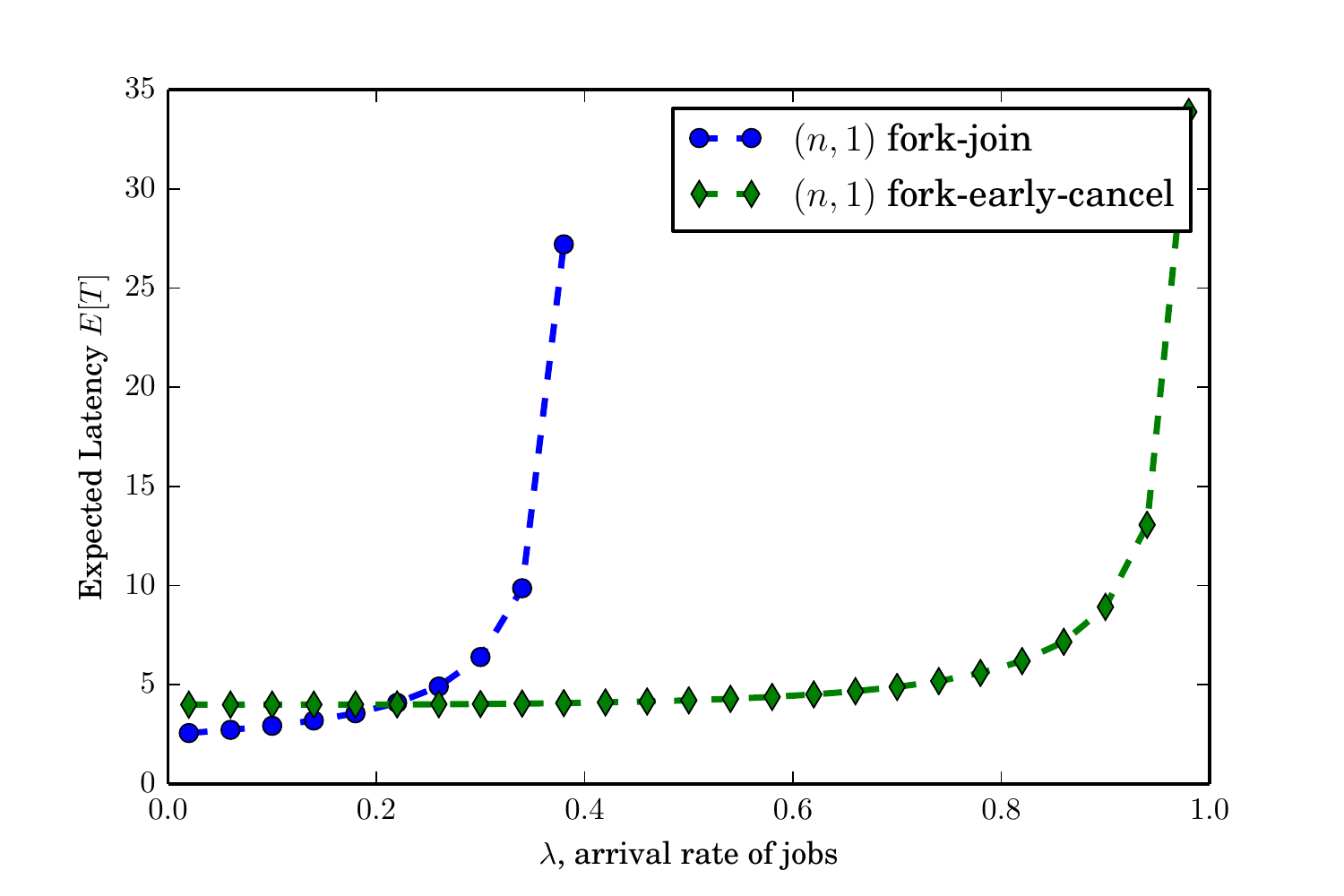}
    \caption{For the $(4,1)$ system with service time $X \sim \SExp(2, 0.5)$ which is log-concave, early cancellation is better in the high $\lambda$ regime, as given by \Cref{coro:early_cancel_E_T_trend}. \label{fig:normal_early_vs_lambda_log_concave}}
 \vspace{-0.25cm}
\end{figure}

%
%

%
%

Fig.~\ref{fig:normal_early_vs_lambda_log_concave} and Fig.~\ref{fig:normal_early_vs_lambda_log_convex} illustrate \Cref{coro:early_cancel_E_T_trend}. Fig.~\ref{fig:normal_early_vs_lambda_log_concave} shows a comparison of $\E{T}$ with and without early cancellation of redundant tasks for the $(4,1)$ system with service time $X \sim \SExp(2,0.5)$. We observe that early cancellation gives lower $\E{T}$ in the high $\lambda$ regime. In Fig.~\ref{fig:normal_early_vs_lambda_log_convex} we observe that when $X$ is $\HyperExp(0.1, 1.5, 0.5)$ which is log-convex, early cancellation is worse for both small and large $\lambda$.

In general, early cancellation is better when $X$ is less random (lower coefficient of variation). For example, a comparison of $\E{T}$ with $(n,1)$ fork-join and $(n,1)$ fork-early-cancel systems as $\Delta$, the constant part of service time $\SExp(\Delta, \mu)$ varies indicates that early cancellation is better for larger $\Delta$. When $\Delta$ is small, there is more randomness in the service time, and hence keeping the redundant tasks running gives more diversity and lower $\E{T}$. But as $\Delta$ increases, task service times are more deterministic and the diversity benefit of having redundant tasks is smaller. 

\begin{figure}[t]
\centering
\includegraphics[width=3.2in]{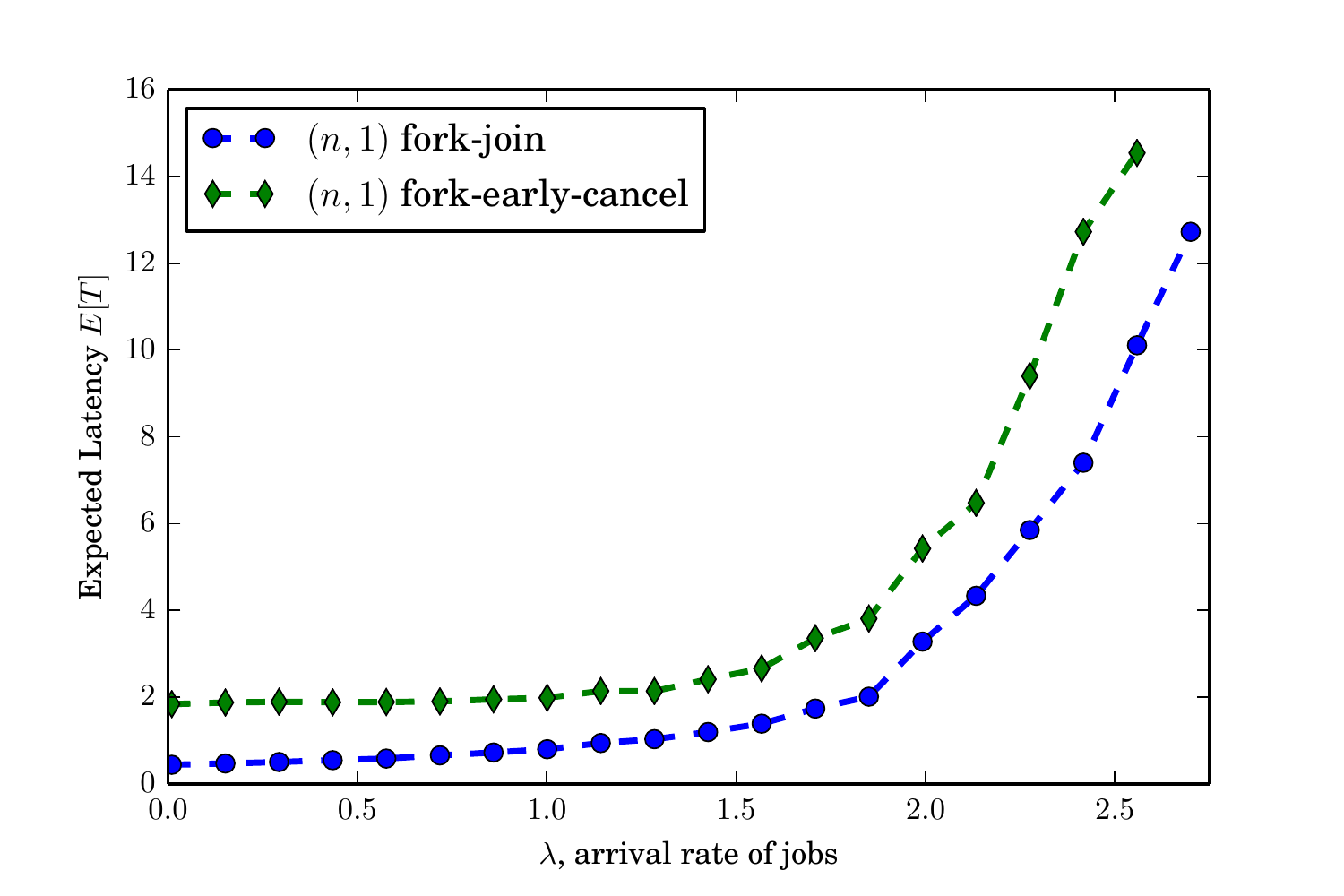}
\caption{ For the $(4,1)$ system with $X \sim \HyperExp(0.1, 1.5, 0.5)$, which is log-convex, early cancellation is worse in both low and high $\lambda$ regimes, as given by \Cref{coro:early_cancel_E_T_trend}. \label{fig:normal_early_vs_lambda_log_convex}}
 \vspace{-0.25cm}
\end{figure}

\section{$(n,r,1)$ partial-fork-join system}
\label{sec:partial_fork}

We now analyze the latency-cost trade-off for the $(n,r,1)$ partial-fork-join system where an incoming job is forked to some $r$ out $n$ servers and we wait for any $1$ task to finish. The $r$ servers are chosen using a symmetric policy (\Cref{defn:symmetric_forking}).  Some examples of symmetric policies are:
\begin{enumerate}
\item \textit{Group-based random: } This policy holds when $r$ divides $n$. The $n$ servers are divided into $n/r$ groups of $r$ servers each. A job is forked to one of these groups, chosen uniformly at random. 
\item \textit{Uniform Random: } A job is forked to any $r$ out of $n$ servers, chosen uniformly at random.\footnote{In the $(n,r,1)$ partial-fork-join with uniform random policy, we replicate the task at $r$ randomly chosen queues. Instead, in the power-of-$r$ scheduling \cite{powerof2}, a job is assigned to the shortest of $r$ randomly chosen queues. Power-of-$r$ is similar to early cancellation of $r-1$ out of $r$ tasks, even before they join the queues. Thus, by \Cref{coro:early_cancel_E_T_trend}, we conjecture that for log-convex $\bar{F}_X$, the $(n,r,1)$ partial-fork-join system gives lower latency than power-of-$r$ scheduling. }
\end{enumerate}

Fig.~\ref{fig:partial_fork} illustrates the $(4,2,1)$ partial-fork-join system with the group-based random and the uniform-random policies. In the sequel, we develop insights into the best $r$ and the choice of servers for a given service distribution $F_X$. 
%

\begin{figure}[t]
\centering
\begin{subfigure}[t]{0.5\linewidth}
    \centering
   \includegraphics[width=1.80in]{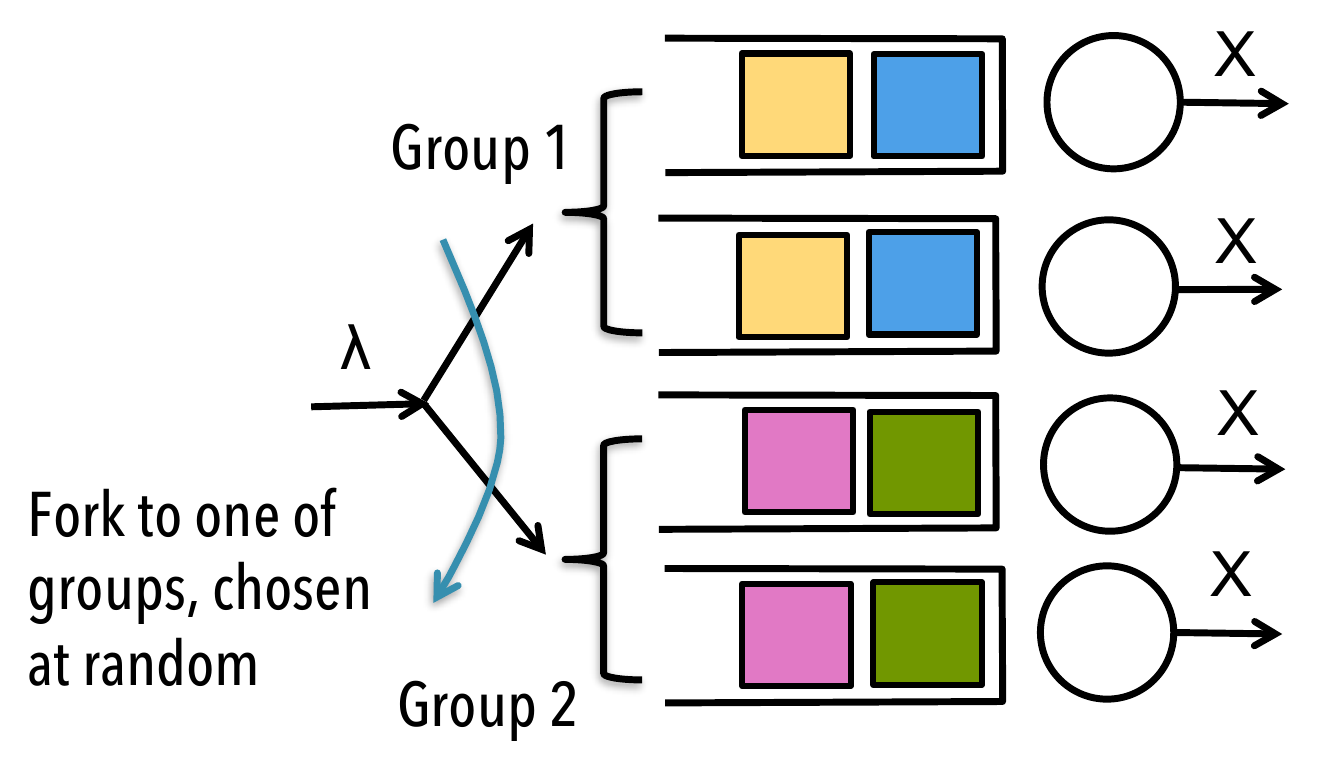}
	\caption{Group-based random}
\end{subfigure}
~
\begin{subfigure}[t]{0.4 \linewidth}
    \centering
   \includegraphics[width=1.55in]{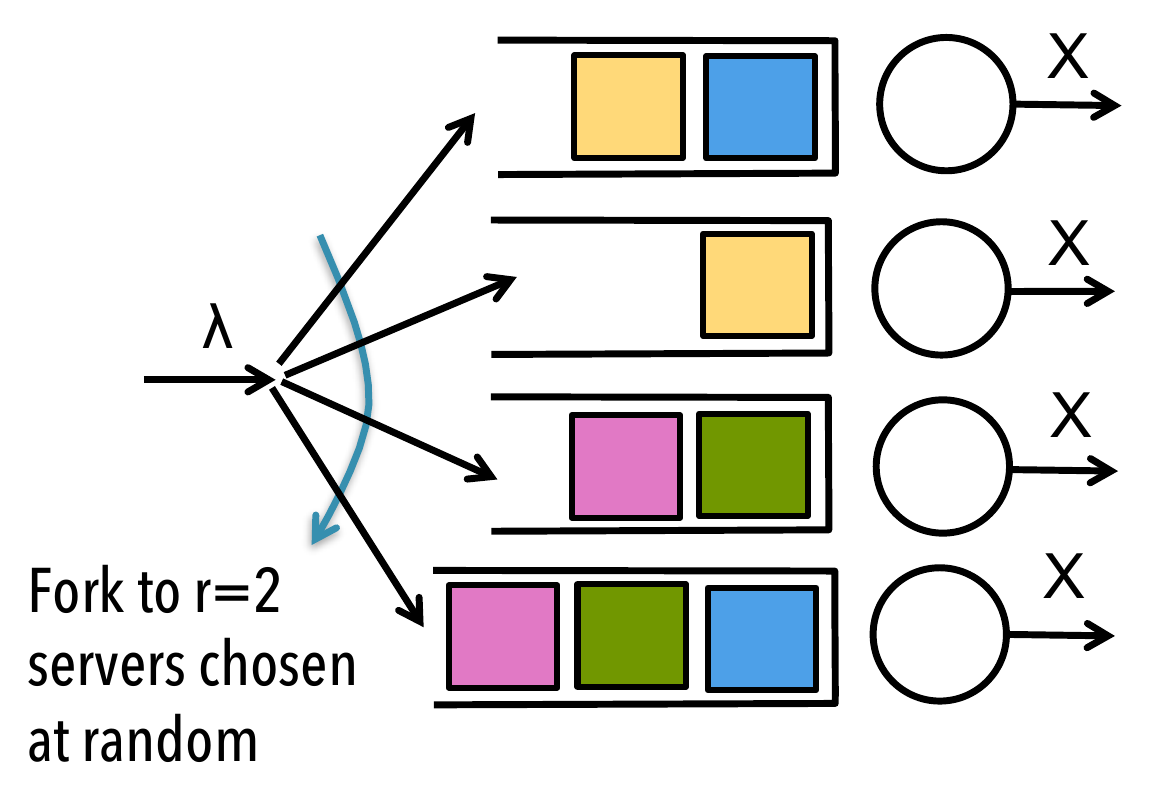}
\caption{Uniform random}
\end{subfigure}
\caption{$(4,2,1)$ partial-fork-join system, where each job is forked to $r=2$ servers, chosen according to the group-based random or uniform random policies. \label{fig:partial_fork}}
\end{figure}

\subsection{Latency-Cost Analysis}
%
In the group-based random policy, each group behaves as an $(r,1)$ fork-join system, the $r$ tasks of a job starting service simultaneously. Thus, the expected latency and cost follow from \Cref{thm:rep_queueing} as given in \Cref{lem:latency_cost_group_based} below.

\begin{lem}[Group-based random]
\label{lem:latency_cost_group_based}
The expected latency and cost when each job is forked to one of $n/r$ groups of $r$ servers each are given by
\begin{align}
\E{T} &= \E{X_{1:r}}  + \frac{\lambda r \E{X_{1:r}^2}}{2(n - \lambda r \E{X_{1:r}})} \label{eqn:E_T_group_based} \\
\E{C} &= r \E{X_{1:r}} \label{eqn:E_C_group_based}
\end{align}
\end{lem}

\begin{proof}
Since the job arrivals are split equally across the $n/r$ groups, such that the arrival rate to each group is a Poisson process with rate $\lambda r/n$. The $r$ tasks of each job start service at their respective servers simultaneously, and thus each group behaves like an independent $(r,1)$ fork-join system with Poisson arrivals at rate $\lambda r/n$. Hence, the expected latency and cost follow from \Cref{thm:rep_queueing}. 
\end{proof}

Using \eqref{eqn:E_C_group_based} and \Cref{clm:capacity_in_terms_of_EC}, we can infer that the service capacity (maximum supported $\lambda$) for an $(n,r,1)$ system with group-based random policy is
\begin{align}
\lambda_{max} = \frac{n}{r\E{X_{1:r}}} \label{eqn:lmbda_max_group_based}
\end{align}
From \eqref{eqn:lmbda_max_group_based} we can infer that the $r$ that minimizes $r \E{X_{1:r}}$ results in the highest service capacity, and hence the lowest $\E{T}$ in the high traffic regime. By \Cref{lem:r_E_X_1_r_trend}, the optimal $r$ is $r=1$ ($r=n$) for log-concave (log-convex) $\bar{F}_X$.

For other symmetric policies, it is difficult to get an exact analysis of $\E{T}$ and $\E{C}$ because the tasks of a job can start at different times. However, we can get bounds on $\E{C}$ depending on the log-concavity of $X$, given in \Cref{thm:E_C_r_trend} below.

\begin{thm}
\label{thm:E_C_r_trend}
Consider an $(n,r,1)$ partial-fork join system, where a job is forked into tasks at $r$ out of $n$ servers chosen according to a symmetric policy. For any relative task start times $t_i$, $\E{C}$ can be bounded as follows.
\begin{align} 
 r \E{X_{1:r}} \geq \E{C} &\geq \E{X} \quad \quad \text{if } \bar{F}_X \text{ is  log-concave} \label{eqn:E_C_bounds_log_concave}\\
\E{X} \geq \E{C} &\geq r \E{X_{1:r}}  \quad \text{if } \bar{F}_X \text{ is  log-convex} \label{eqn:E_C_bounds_log_convex}
\end{align} 
In the extreme case when $r=1$, $\E{C} = \E{X}$, and when $r = n$, $\E{C} = n \E{X_{1:n}}$. 
\end{thm}

To prove \Cref{thm:E_C_r_trend} we take expectation on both sides in \eqref{eqn:C_expr}, and show that for log-concave and log-convex $\bar{F}_X$, we get the bounds in \eqref{eqn:E_C_bounds_log_concave} and \eqref{eqn:E_C_bounds_log_convex}, which are independent of the relative task start times $t_i$. The detailed proof is omitted here, but can be found in the extended version \cite{gauri_tompecs_arxiv_2015}.

In the sequel, we use the bounds in \Cref{thm:E_C_r_trend} to gain insights into choosing the best $r$ and best scheduling policy when $\bar{F}_X$ is log-concave or log-convex.

\subsection{Optimal value of $r$}

\begin{figure}[t]
    \centering
    \includegraphics[width=3.25in]{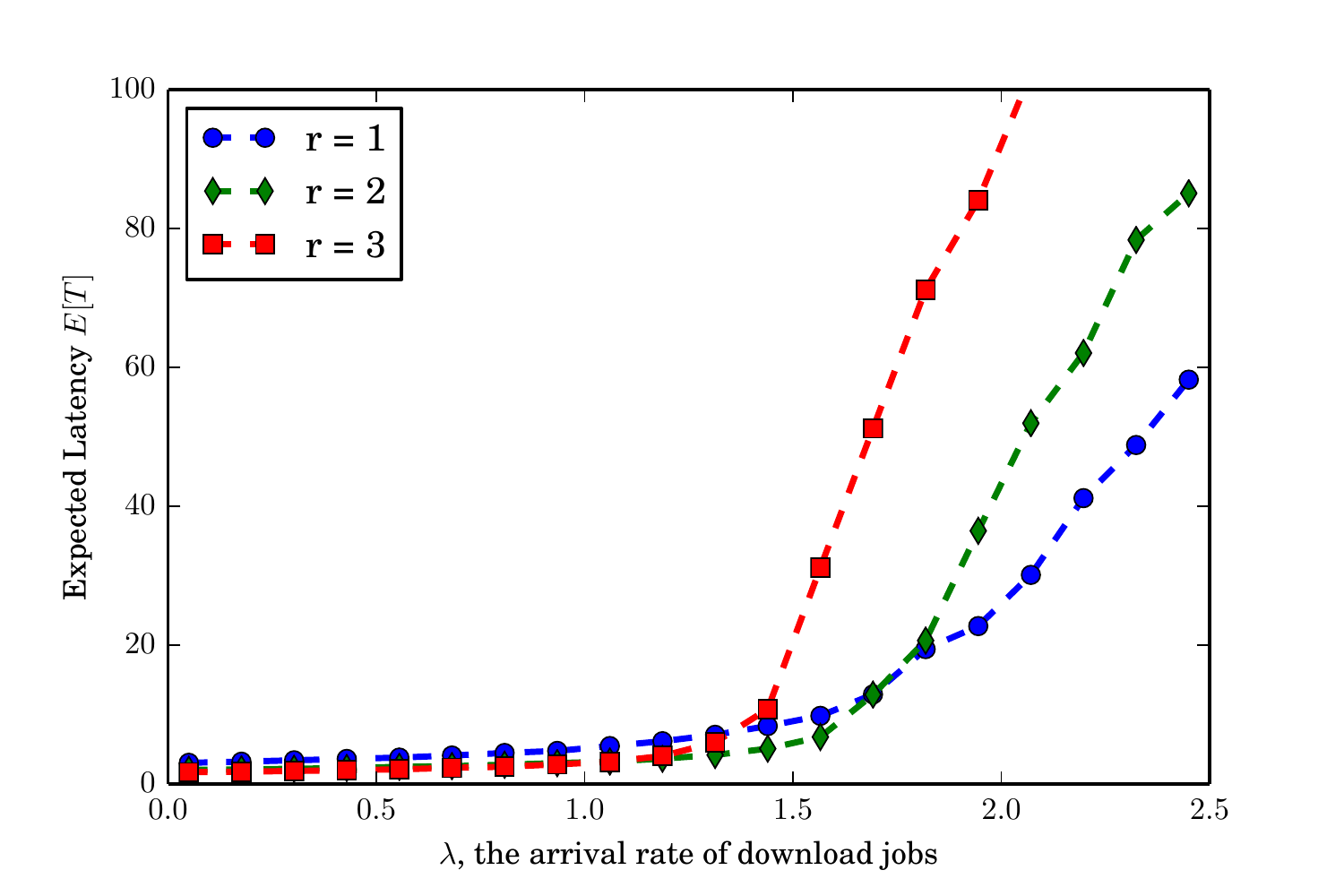}
    \caption{For $X \sim \SExp(1,0.5)$ which is log-concave, forking to less (more) servers reduces expected latency in the low (high) $\lambda$ regime. \label{fig:E_T_rep_shifted_exp_vs_lambda_diff_r}}
 \vspace{-0.25cm}
\end{figure}

By \Cref{lem:r_E_X_1_r_trend}, $r \E{X_{1:r}}$ is non-decreasing (non-increasing) with $r$ for log-concave (log-convex) $\bar{F}_X$. By this fact and \Cref{thm:E_C_r_trend}, we get the following corollaries about how $\E{C}$ and $\E{T}$ vary with $r$. 

\begin{coro}[Expected Cost vs. $r$]
\label{coro:EC_vs_r}
For a symmetric policy, forking of each job to $r$ out of $n$ servers, $r=1$ ($r=n$) minimizes the expected cost $\E{C}$ when $\bar{F}_X$ is log-concave (log-convex).
\end{coro}
%
%
%
%

\begin{coro}[Expected Latency vs. $r$]
\label{coro:latency_vs_r}
In the low-traffic regime, forking to all servers ($r =n$) gives the lowest $\E{T}$ in the low $\lambda$ regime for any service time distribution $F_X$. In the high traffic regime, $r=1$ ($r=n$) gives lowest $\E{T}$ if $\bar{F}_X$ is log-concave (log-convex).
\end{coro}

\begin{figure}[t]
\centering
\includegraphics[width=3.25in]{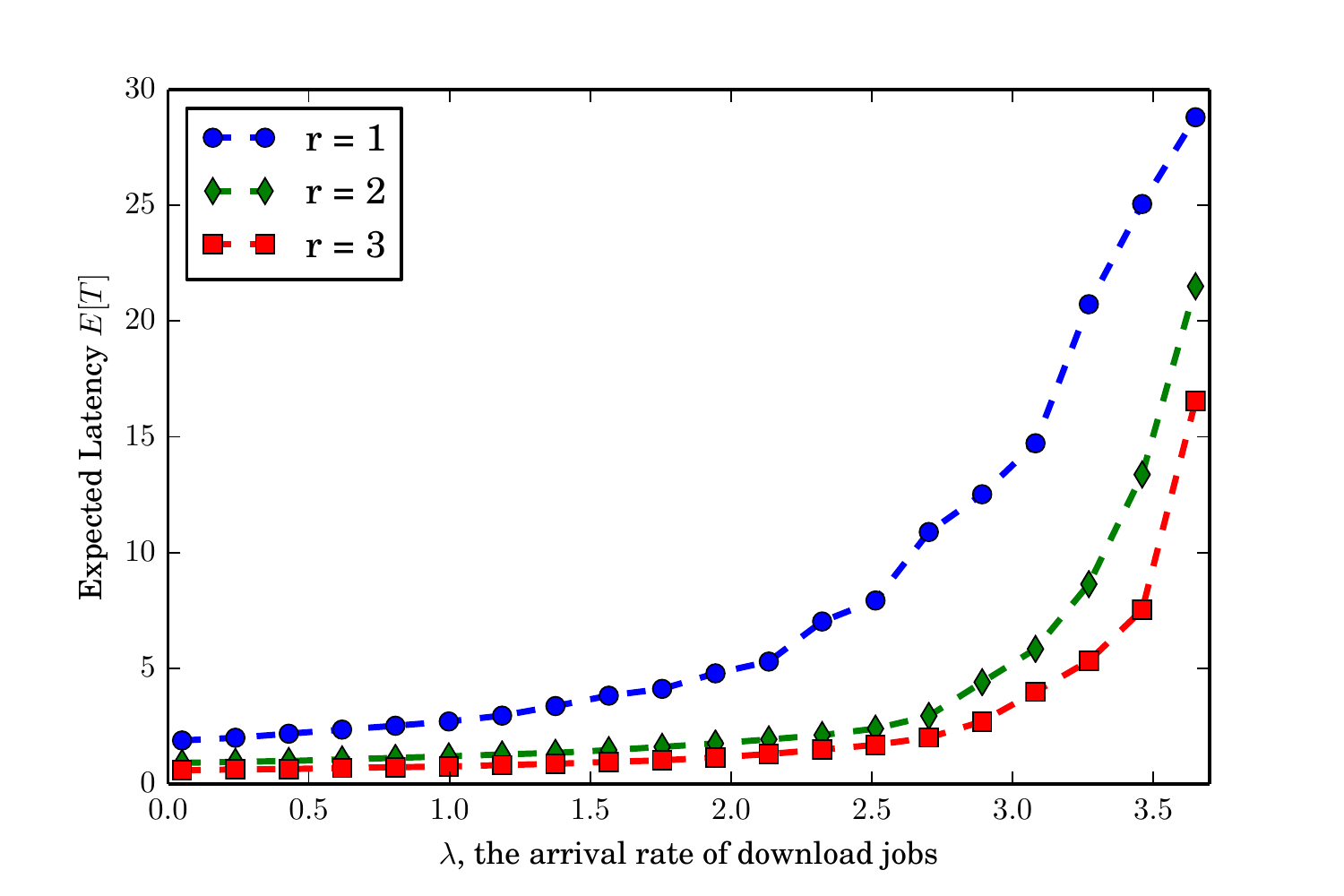}
\caption{For $X \sim \HyperExp(p, \mu_1, \mu_2)$ with $p=0.1$, $\mu_1 = 1.5$, and $\mu_2 = 0.5$ which is log-convex, forking to more servers (larger $r$) gives lower expected latency for all $\lambda$. \label{fig:E_T_rep_hyper_exp_vs_lambda_diff_r}}
 \vspace{-0.25cm}
\end{figure}

\Cref{coro:latency_vs_r} is illustrated by Fig.~\ref{fig:E_T_rep_shifted_exp_vs_lambda_diff_r} and Fig.~\ref{fig:E_T_rep_hyper_exp_vs_lambda_diff_r} where $\E{T}$ is plotted versus $\lambda$ for different values of $r$. Each job is assigned to $r$ servers chosen uniformly at random from $n = 6$ servers. In \Cref{fig:E_T_rep_shifted_exp_vs_lambda_diff_r} the service time distribution is $\SExp(\Delta, \mu)$ (which is log-concave) with $\Delta = 1$ and $\mu = 0.5$. When $\lambda$ is small, more redundancy (higher $r$) gives lower $\E{T}$, but in the high $\lambda$ regime, $r=1$ gives lowest $\E{T}$ and highest service capacity. On the other hand in Fig.~\ref{fig:E_T_rep_hyper_exp_vs_lambda_diff_r}, for a log-convex distribution $\HyperExp(p, \mu_1, \mu_2)$, in the high load regime $\E{T}$ decreases as $r$ increases.

\Cref{coro:latency_vs_r} was previously proven for new-better-than-used (new-worse-than-used) instead of log-concave (log-convex) $\bar{F}_X$ in \cite{shah_when_2013,koole_righter_2008}, using a combinatorial argument. Using \Cref{thm:E_C_r_trend}, we get an alternative, and arguably simpler way to prove this result. Note that our version is slightly weaker because log-concavity implies new-better-than-used but the converse is not true in general. 

\subsection{Choice of the $r$ servers}
For a given $r$, we now compare different policies of choosing the $r$ servers for each job. The choice of the $r$ servers determines the relative starting times of the tasks. If all the $r$ tasks start at the same time, $\E{C} = r \E{X_{1:r}}$. By comparing with the bounds in \Cref{thm:E_C_r_trend} that hold for any relative task start times we get the following result. 


\begin{coro}[Cost for different policies]
\label{coro:cost_diff_policies}
Given $r$, if $\bar{F}_X$ is log-concave (log-convex), the symmetric policy that results in the tasks starting at the same time ($t_i = 0$ for all $1 \leq i \leq r$) results in higher (lower) $\E{C}$ than one that results in $0 < t_i < \infty$ for some $i$.
\end{coro}

\begin{coro}[Latency in high $\lambda$ regime]
\label{coro:latency_diff_policies}
Given $r$, if $\bar{F}_X$ is log-concave (log-convex), the symmetric policy that results in the tasks starting at the same time ($t_i = 0$ for all $1 \leq i \leq r$) results in higher (lower) $\E{T}$ in the high traffic regime than one that results in $0 < t_i < \infty$ for some $i$.
\end{coro}

For example, lets us compare the group-based random and uniform random policies. The $r$ tasks may start at different times with the uniform random policy, whereas they always start simultaneously with group-based random policy. Thus, in the high $\lambda$ regime, that uniform random policy results lower latency for log-concave $\bar{F}_X$. But for log-convex $\bar{F}_X$, group-based forking is better in the high $\lambda$ regime.

\section{Concluding Remarks}
\label{sec:conclu}
We consider a redundancy model where a computing job is replicated at multiple servers, and we wait for any one copy to finish and cancel the rest. We analyze how redundancy affects the latency, and the cost of computing time, and demonstrate how the log-concavity of service time is a key factor in determining the best redundancy strategy. For example, if the service time is log-convex, adding maximum redundancy reduces both latency and cost. For log-concave service time, can reduce latency, but increases the cost of computing time. Thus, adding fewer replicas, and canceling redundant tasks early is more effective, especially in the high traffic regime.

Using these insights, in \cite{gauri_tompecs_arxiv_2015}, we propose a general redundancy strategy for an arbitrary service time distribution, that may be neither log-concave nor log-convex. Ongoing work includes developing online strategies to simultaneously learn the service distribution, and the best redundancy strategy. More broadly, the proposed redundancy techniques can be used to reduce latency in several applications beyond the realm of cloud storage and computing systems, for example crowdsourcing, algorithmic trading, manufacturing etc.\

\section{Acknowledgments}
We thank Devavrat Shah, Da Wang, Sem Borst and Rhonda Righter for helpful discussions during this work. 


\bibliographystyle{ieeetr}
\bibliography{storage,computing}

\end{document}